\RequirePackage{amsmath}
\pdfoutput=1
\documentclass{lms}
\input{glyphtounicode}
\usepackage{lmodern}
\usepackage[T1]{fontenc}
\usepackage[utf8]{inputenc}
\usepackage{amsfonts}
\usepackage{url}
\usepackage[english]{babel}
\usepackage{algorithm}
\usepackage{algpseudocode}
\usepackage{graphicx}
\usepackage{amssymb}

\usepackage{amsthm}

\usepackage{tikz}

\newcommand{\F}{\mathbb{F}}
\newcommand{\R}{\mathbb{R}}
\newcommand{\Z}{\mathbb{Z}}
\newcommand{\Q}{\mathbb{Q}}
\newcommand{\C}{\mathbb{C}}
\newcommand{\Ord}{\mathcal{O}}

\mathchardef\mhyphen="2D 

\newcommand{\pr}[2]{\langle #1,#2 \rangle}

\newcommand{\Dis}{\mathcal{D}}

\providecommand{\tr}{\mathop{\rm tr}\nolimits}
\newcommand{\Mod}{\text{ mod }}
\newcommand{\Frk}[1]{\mathfrak{#1}}
\newcommand{\Cal}[1]{\mathcal{#1}}
\newcommand{\Norm}{\mathcal{N}}
\newcommand{\tens}{\otimes}
\newcommand{\barre}{\overline}
\DeclareMathOperator*{\argmax}{arg\,max}
\DeclareMathOperator*{\argmin}{arg\,min}
\usepackage[pdftex,bookmarks,bookmarksopen,bookmarksdepth=3]{hyperref}
\newtheorem{theorem}{Theorem}
\newtheorem{lemma}{Lemma}

\newnumbered{definition}{Definition}
\newnumbered{notation}{Notation}
\newnumbered{problem}{Problem}
\newnumbered{heuristic}{Heuristic}

\newcommand{\recalltheorem}[2]{
}

\title{Algorithms on Ideal over Complex Multiplication order}
\author{Paul Kirchner}

\begin{document}
\maketitle

\begin{abstract}
We show in this paper that the Gentry-Szydlo algorithm for cyclotomic orders, previously revisited by Lenstra-Silverberg, can be extended to complex-multiplication (CM) orders, and even to a more general structure.
This algorithm allows to test equality over the polarized ideal class group, and finds a generator of the polarized ideal in polynomial time.
Also, the algorithm allows to solve the norm equation over CM orders and the recent reduction of principal ideals to the real suborder can also be performed in polynomial time.
Furthermore, we can also compute in polynomial time a unit of an order of any number field given a (not very precise) approximation of it.

Our description of the Gentry-Szydlo algorithm is different from the original and Lenstra-Silverberg's variant and we hope the simplifications made will allow a deeper understanding.

Finally, we show that the well-known speed-up for enumeration and sieve algorithms for ideal lattices over power of two cyclotomics can be generalized to any number field with many roots of unity.
\end{abstract}

\section{Introduction}

Recently, an algorithmic study of lattices was made necessary by new cryptographic proposals.
Indeed, lattice-based cryptography has several advantages : it seems post-quantum secure, allows to build a lot of cryptosystems and enjoys a worst-case/average-case reduction over any lattice problem~\cite{STOC:Regev05}.
Yet, the schemes are slow and have large keys so that most designers turned towards {\it ideal}-lattice based cryptography, based on the Ring-LWE~\cite{EC:LyuPeiReg10}.
However, ideal lattices are less studied and the added algebraic structure might allow significant gains with respect to the same problem for a random lattice.
Some suggested to remove a part of the algebraic structure by choosing a polynomial with a large Galois group~\footnote{\url{http://blog.cr.yp.to/20140213-ideal.html}}.

We show that many algorithms previously discovered for cyclotomic fields can be generalized to CM orders with little loss. 
In particular, a large part of the present paper is dedicated to the Gentry-Szydlo algorithm~\cite{EC:GenSzy02}.
This algorithm, given an ideal generated by some $v\in \Z[X]/(X^n-1)$ and the autocorrelation of $v$, finds $v$ up to the (few) root of unity in polynomial time.
It was first used to break the NTRU signature scheme~\cite{EC:GenSzy02}.
Lenstra and Silverberg~\cite{C:LenSil14} then extended it to (essentially) product of cyclotomic rings, and made it rigorous.
According to~\cite{C:LenSil14}, Gentry referred to Gentry-Szydlo's algorithm as "a rather crazy, unusual combination of LLL with more `algebraic' techniques", while Smart viewed it as magic.
We hope the simplifications made to the algorithms\footnote{In a recent conference, no less than 9 hours and 33 minutes were dedicated to Gentry-Szydlo's algorithm and Lenstra-Silverberg's modifications.}, as well as our effort to see the exact conditions under which Gentry-Szydlo's algorithm can be run will help to removed this "dark magic" aspect.

This new algebraic structure extends the CM order, and can be applied to any order.
We give here a particular case of our main theorem~\autoref{thm:gsalg} :
\begin{theorem}
Under GRH or using randomness, we can test in polynomial time if $(I,r)$ where $I$ is an ideal of a CM order is equal to some $((v),v\barre v)$ for some invertible $v$ and find such a $v$.
\end{theorem}

It has several applications, that we extend in \autoref{sec:appli}.
The most useful is the reduction of searching for short vectors in ideals over CM orders to ideals over their real suborder, which is enough to attack the Smart-Vercauteren FHE scheme~\cite{PKC:SmaVer10,EC:GenHal11}, the GGH multilinear map scheme~\cite{EC:GarGenHal13}, Soliloquy~\cite{campbell2014soliloquy} and the Boneh-Freeman homomorphic signatures~\cite{EC:BonFre11}.
Remark that a quantum algorithm was recently discovered which breaks all these schemes in polynomial time~\cite{SODA:BiaSon16}.
This problem can be mitigated by switching to other number fields, since the attack crucially relies on having a very orthogonal basis of the unit lattice~\cite{EPRINT:CDPR15}.
It (usually) divides by two the dimension of the lattice, which corresponds to reduce the running time to its square root.
For most cryptosystems based on ideal lattices, the ideals are \textit{not} principal, so this attack does not apply~\footnote{Remark that the schemes are proven to be at least as secure as ideal lattices, current attacks need to find short vectors over a module of rank two.}.
However, if it can be modified to attack Ring-LWE, which is decoding in the lattice $\begin{pmatrix} q & a \\ 0 & 1\end{pmatrix}\Ord^2$ for some uniform $a\in \Ord/q$, this would break almost all published design in practical time~\footnote{We consider therefore that a wise design based on ideal lattices should have at least 256 bits of security if the users want to keep data secure for several decades with high probability.}.
Another application is to solve the norm equation, though not in polynomial time.
A last application is the heuristic ability to solve bounded distance decoding in polynomial time over the unit lattice with approximation factor around $n^{O(\log(\log n))}$, even though the unit lattice is not known.
Note however that units may not be of size polynomial in the discriminant.
A recent attack against GGH uses Gentry-Szydlo's algorithm in a similar way, our extension allows it to work over any number field~\cite{albrechtsubfield}.

While CM orders and their polarized class group were introduced for the study of abelian varieties, such as some elliptic curves, this paper does not use any algebraic geometry.
Also, this result may indicate that polarized class group are more tractable than the class group from a computational point of view.
In particular, given some generators of a subgroup of the polarized class group, either we have an incremental multiset hash function~\cite{AC:CDDGS03} (or a hash function), or there is an efficient bijection between this abelian group and its normal form.
Finally, testing equality in the polarized class group is related to isotopy of knots, see~\cite[Knot theory]{bayer2002ideal}.

In appendix, we show that for ideal lattices, we can accelerate standard algorithms for searching the shortest vector by a polynomial factor with respect to their (exponential) complexity on generic lattices.

\section{Preliminaries}\label{sec:preliminaries}

\subsection{Basics}

The norm of a vector or a matrix $||A||$ is the Frobenius form, i.e. $||A||^2=\sum_{i,j} A_{i,j}^2$.
The binary logarithm is denoted by $\log$ and $\ln$ is the neperian logarithm.
All indices start from zero.

\begin{definition}
An order $\Ord$ is a commutative unitary ring, whose additive group is isomorphic to $\Z^n$.
We denote $\Ord^*$ the group of units in the order.
The trace of $a$, $\tr(a)$, is the trace of the endomorphism $x\mapsto ax$.
\end{definition}
Notice that usually, an order is defined with respect to an algebra, this is not the case here.

For the rest of the paper, we will work with an order, always denoted by $\Ord$, whose corresponding algebra $\Q \tens \Ord$ is denoted by $E$.
We refer to elements of $E^* \cap \Ord$ as \textit{invertible}.
An order is given to an algorithm by $n$, followed by the $n(n+1)/2$ products of two basis elements, written as integer linear combinations of basis elements.
We define the height function of $x\in \Z$ to be $h(x)=2\log(2+|x|)$, and extend it to the rationals $h(p/q)=h(p)+h(q)$ where $p$ and $q$ are coprime.
The height of a matrix or a list is the sum of the heights of the components.
It represents the number of bits used to describe an object.

\begin{definition} An integral ideal $I$ over an order $\Ord$ is an $\Ord$-submodule.
Its norm $N(I)$ is defined as the cardinal of $\Ord/I$.
A (fractional) ideal is an integral ideal, up to a rational number.
An ideal is said to be \textit{invertible} if there exists a fractional ideal $J$ such that $IJ=\Ord$.
\end{definition}

An ideal is described by a $\Z$-basis.
It is well known that ideals can be multiplied in polynomial time~\cite[Section 4.7]{cohen2013course}.
It is clear that if $v\in \Ord$ is invertible, then $(v)$ is invertible.

\begin{definition} A positive-definite quadratic form $f:E\to \Q$ over a $\Q$ vector space $E$ is a function such that $(a,b)\mapsto (f(a+b)-f(a-b))/4$ is bilinear, and $f(x)>0$ for all non-zero $x$.
Its determinant is the determinant of the corresponding bilinear form.
\end{definition}
While most algorithms in the litterature are presented with lattices, they can usually be transformed into algorithms over quadratic forms and we will be forced here to use the quadratic form.

\begin{notation}
The group of the roots of units in an order $\Ord$ is denoted $\mu(\Ord)$.
\end{notation}
\begin{notation}
The $p$-Sylow of a finite abelian group $G$ is denoted by $G_p$.
\end{notation}

\begin{notation}
Given a number field $F$, we denote its ring of integers by $\Z_F$ and its discriminant by $d_F$.
\end{notation}

\begin{theorem}\label{thm:structalg} Let $E$ be a commutative algebra of dimension $n$ over the field $K$.
Then there is a unique $E'=\prod_i F_i$ where $F_i$ is a finite extension of $K$ such that for all $x\in E$, there is a unique $(a,b)$ with $a$ nilpotent, $a+b=x$ and $b\in E'$.
\end{theorem}
\begin{proof}
This is a direct consequence of the Artin-Wedderburn theorem.
\end{proof}
On first read, the reader may assume that $\Ord=\Z[\alpha]$ for some algebraic $\alpha$.

An algorithm has a negligible probability of failure if the probability of failure is bounded by $2^{-\Omega(n)}$.

\subsection{Advanced definitions}

\begin{definition} A split CM order is an order $\Ord$, a norm function $N:E\to\Norm$ where $\Norm$ is a commutative semigroup, and a trace function $\tr:\Norm\to \Q$ such that :
\begin{itemize}
\item $N$ is a morphism as a semigroup, i.e. $N(xy)=N(x)N(y)$ for all $x,y \in E$
\item $x\mapsto \tr(N(x))$ is a positive-definite quadratic form.
\end{itemize}
\end{definition}
Without loss of generality, we will impose furthermore that $(\tr(N(x+y))-\tr(N(x-y)))/4$ is an integer for all $x,y\in \Ord$.

Norms will be used in our algorithms in a black-box manner, so we define the height of a norm to be the number of bits used to represent it.
We describe a split CM order with functions which run in polynomial time of the height of the input and $h(\Ord)$ for $N$, multiplication and inversion in $\Norm$, as well as the trace.

We now give several examples of split CM orders.
The simplest one, not interesting for our purpose, is where the norm is the identity function and the trace is any positive-definite quadratic form.

\begin{definition} A CM order is an order $\Ord$ equipped with an automorphism $x\mapsto \barre{x}$ which is an involution, and such that $\tr(x\barre{x})$ is a positive-definite quadratic form.
The real suborder $\Ord^+$ is defined as $\{x=\barre{x};x \in \Ord\}$, the imaginary lattice $\Ord^-$ is $\{x=-\barre{x};x \in \Ord\}$.
\end{definition}

We can easily build a split CM order from a CM order by taking $N(x)=x\barre{x}$ and the trace function as the standard trace.
For $\alpha \in \C$ some algebraic number, if $\barre{\alpha} \in \Z[\alpha]$, then $\Z[\alpha]$ equipped with the conjugation is a CM order.
In particular, $\Z[\zeta_n]$ where $\zeta_n=\exp(2i\pi/n)$ and $n$ a positive integer is a CM order.
In these cases, the norm function is the \textit{algebraic norm} over the real subfield, which must not be confused with the corresponding \textit{geometric norm}.
We define a cyclotomic field to be any $\Q[\zeta_n]$.
Also, $\Z[X]/(X^n-1)$ equipped with $X\mapsto X^{n-1}$ is a CM order and the norm is the autocorrelation.

Finally, $\Z[\alpha]$ with the norm $\alpha \mapsto \alpha\barre{\alpha}$ is a {\it split} CM order.
We can generalize this construction using~\autoref{thm:structalg} to give a non-trivial split CM order from any order.

\begin{definition}
A \textit{polarized ideal} is a pair $(I,r)$ where $I$ is a fractional invertible ideal and $r\in \Norm^*$.
The determinant of a polarized ideal $(I,r)$ is the determinant of $x\mapsto \tr(N(x)/r)$ defined over $I$.
The \textit{polarized ideal group} is the group of all polarized ideals.
The \textit{principal} polarized ideal group is the group of all $((v),N(v))$ for all $v\in E^*$, and is a subgroup of the polarized ideal group.
\end{definition}

We can now state the \textit{informal} problem we (partly) solve :
\begin{problem}
Given a split CM order and a polarized ideal $(I,r)$, determine all invertible solutions $v\in \Ord$ to $(I,r)=((v),N(v))$.
\end{problem}

\begin{theorem}~\label{thm:racineunite}
If $N(v)=N(w)$ for $v,w\in \Ord$ invertibles, then $v/w\in\mu(\Ord)$.
\end{theorem}
\begin{proof}
$N((v/w)^k)=1$ for all integer $k$, but since $x\mapsto \tr(N(x))$ is definite positive, $(v/w)^k$ can take only a bounded number of values.
Hence, $v/w$ is a root of unity.
\end{proof}
This implies that $v$ is defined up to a group of roots of unity.

Though we will not need the following definition, it gives a nice interpretation of the Gentry-Szydlo algorithm.
\begin{definition}
Let $\Delta$ be the determinant of $(\Ord,N(1))$.
The polarized ideal class group is the maximum group of polarized ideals $(I,r)$ of determinant $\Delta$, modulo the principal polarized ideal group.
\end{definition}

\begin{lemma}\label{lem:cmnilpot}
The only nilpotent of a CM order $\Ord$ is $0$.
\end{lemma}
\begin{proof}
Let $x$ be a nilpotent.
Then $x\barre{x}$ is also a nilpotent, so that $\tr(x\barre x)=0$.
Since $x\mapsto \tr(x\barre x)$ is definite, $x=0$.
\end{proof}

The following definition is the main conceptual novelty with respect to previous descriptions of Gentry-Szydlo's algorithm, and allows to greatly simplify its exposition.
\begin{definition}
The formal products over $E$, $\Z^{(E)}$, is the additive group of function from $E$ to $\Z$ which are non-zero on a finite set, called support.
The evaluation of $f\in\Z^{(E)}$, denoted by $[f]$, is defined by \[ \prod_{x\in E} x^{f(x)}.\]
We call the elements of the support the base, and $f(x)$ is called the exponent of $x$.
The group law is denoted multiplicatively.
\end{definition}
For clarity reasons, we will consider an element $x\in E$ to be also the function of $\Z^{(E)}$ equal to zero everywhere except in $x$ where it is equal to one.

\begin{definition}
We call a group of polarized ideals to be \textit{reducible} if for all $x$, $\tr^{-1}(\{x\})$ is finite.

We call a group of polarized ideals to be \textit{poly-reducible} if there is a constant $c$ such that for all $(I,r)$ in the group with $I$ integral, $h(r)\leq (h(\Ord)+\log(\tr(r))^c$. 
\end{definition}
It can be checked that in all our examples of split CM orders, the full group is poly-reducible.
From a mathematical perspective, reducibility is more pertinent, as it is enough to prove that a reducible group modulo the principal polarized ideals is finite.
Yet, we need the stronger condition for efficient computation over this finite group, see~\autoref{sub:power}.
However, if we extend the norms to $\Norm \times \Q$, the polarized ideal class group is clearly infinite.

\begin{lemma}
The principal polarized ideal group is poly-reducible.
\end{lemma}
\begin{proof}
Let $G$ be the Gram matrix of $x\mapsto \tr(N(x))$ over $\Ord$.
Let $G=L^tL$ be its corresponding Cholesky decomposition, that is $L$ is an upper triangular matrix.
Then, we can write $L_{i,i}^2$ as a quotient of two non-zero integer determinants of submatrices of $G$, so that the Hadamard bound gives $L_{i,i}\geq ||G||^{-n/2}$.
Therefore, for any $v\in \Ord$ which is invertible, there exists $i$ such that 
\[ |v_i|^2\geq v^tGv||G||^{-1/n}=\tr(N(v))||G||^{-n}. \]

This implies that $h(v)\geq h(\tr(N(v)))-h(||G||^n)$.
Now, $x\mapsto \tr(N(x))$ is computed in time polynomial in $h(\Ord)+h(x)$, so there exists $c_0$ such that $h(||G||^n)\leq h(O)^{c_0}$.
Also, there exists $c_1$ such that $h(1/N(v))\leq h(v)^{c_1}$.
We conclude that there exists $c_2$ such that $h(N(1/v))\leq (h(\Ord)+h(\tr(N(v))))^{c_2}$.
\end{proof}
\begin{lemma}\label{lem:polyred}
A group is poly-reducible if and only if, there exists a constant $c$ such that for all polarized ideals $(I,r)$ in the group with $\Ord \subset I$, $h(r)\leq(h(\Ord)+\log(\tr(1/r))^c$.
\end{lemma}
\begin{proof}
Remark that $I \subset \Ord$, $1\in I$, and $I^{-1}$ is integral are equivalent.
Thus, because inversion runs in polynomial time, $h(1/r)=h(r)^{\Theta(1)}$ and the result follows.
\end{proof}

\section{Gentry-Szydlo algorithm}

We recall that an order $\Ord$ has a corresponding algebra $\Q \tens \Ord$ denoted by $E$, which contains a maximal product of number fields $E'$.
This section is devoted to prove the following theorem~:
\begin{theorem}\label{thm:gsalg}
Given a polarized ideal $(I,r)$ of a split CM order, if $E'$ is a product of cyclotomic field, or the Generalized Riemann Hypothesis (GRH) is true, or we have access to randomness, we can find $v\omega\in \mu(E)\Ord$ such that $I=(v)$ or prove there is no such $v\omega$ with $(I,r)=((v),N(v))$.
We can also do this unconditionnally in time $2^{n^{(1+o(1))/\log \log n}}$.
Furthermore, if $\Ord$ is a CM order, we can find $v\in \Ord$ such that $(I,r)=(v,N(v))$ under the same conditions.
\end{theorem}

In the next subsections, all algorithms are authorized to fail if there is no solution to $(I,r)=((v),N(v))$.
Since the output of the algorithm can easily be checked, we may assume that there is some solution $v$, except for the analysis of the complexity.
Also, in the case of a CM order, we can assume $v$ is invertible by working in $\Ord/(\Q \tens I)$.

We fix $\Delta$ to be the determinant of $(\Ord,N(1))$.
Remark that $h(\Delta)=h(\Ord)^{O(1)}$.

The hero of our story will be the following group :
\begin{definition}
Given a poly-reducible group $G$ where all $(I,r)\in G$ are of determinant $\Delta$ , its \textit{compactification} is the group of all $(I,r,s)$ where $s\in \Z^{(E)}$ such that $[s]$ is invertible ; modulo the subgroup of all $((1/x),N(1/x),x)$ for all $x\in E^*$.
\end{definition}
Indeed, it replaces Gentry-Szydlo's cumbersome "polynomial chains", and Lenstra-Silverberg's chain of tensor multiplication maps.
Compactification is to be understood in its computer science meaning, i.e. a short representation, and not in a topological sense.

In~\autoref{sub:reduc}, we show that using LLL, we can reduce the description of $I$ and $r$ to a polynomial value, independent of $I$ and $r$.
In~\autoref{sub:power}, we show how to compute a power of $(I,r,1)$ in polynomial time.
We then use this powering algorithm in~\autoref{sub:highpow} to compute the image over a field of $E'$ of some high power of $v$, as a formal product.
By combining various high powers, we show in~\autoref{sub:field} how to compute the image over a field of $v$, up to a root of unity.
Finally, we explain in~\autoref{sub:final} how to compute the nilpotent part of $v$ and one root of unity.

The four different cases in the theorem are introduced in~\autoref{sub:field}.

\subsection{Reduction} \label{sub:reduc}

\begin{theorem}\label{thm:LLLred}
Given a positive-definite integer matrix $G$ of dimension $n$, we can compute in polynomial time a unimodular integer matrix $U$ such that $U^tGU$ has entries bounded by $n2^n\det(G)$.
\end{theorem}
\begin{proof}
See~\cite{lenstra1982factoring} where the Gram-Schmidt orthogonalization is replaced by Cholesky decomposition.
Its output verifies $U^tGU=L^tL$ where $L$ is an upper-triangular matrix, such that $L_{i+1,i+1}\geq L_{i,i}/\sqrt{2}$ and $|L_{i,j}|\leq L_{j,j}$ for all $i,j$.
Let $m=\argmax_i{L_{i,i}}$.
Then, $\prod_{i<m}L_{i,i}^2$ is a positive integer, since it is the determinant of the corresponding upper-left submatrix of $U^tGU$.
Also, for any $j\geq m$, $L_{j,j}\geq L_{m,m}2^{(m-j)/2}$.
It implies that
\[ \Delta=\prod_i L_{i,i}^2\geq L_{m,m}^{2(n-m)}2^{-(n-m)(n-m-1)/2}. \]
We deduce that $L_{m,m}\leq 2^{(n-m-1)/2}\Delta^{1/(2(n-m))} \leq 2^{n/2}\sqrt{\Delta}$.
Using $|L_{i,j}|\leq L_{j,j}$ gives the result.
\end{proof}

\begin{lemma}\label{lem:invert}
Given $m$ matrices $A_i$ in $M_n(\Q)$ such that a linear combination is invertible, we can find $x_i \in \Z$ and $|x_i|\leq n$ such that $\sum_i A_ix_i$ is invertible in polynomial time.
\end{lemma}
\begin{proof}
If $m=1$, we output $x_0=1$.
Else, we compute $r$ the rank of $A_0$ and a set $I$ of $r$ rows and $J$ a set of $r$ columns such that the restriction of $A_0$ to these lines and columns is invertible.
We then recursively find the $x_i$, $i\geq 1$ corresponding to $B_i$, the restriction of $A_i$ to the complements of $I$ and $J$.
Finally, we search through all $x_0$ from $0$ to $n$ and output the first solution.

Without loss of generality, we may analyze this algorithm by assuming $A_0$ is diagonal, with $r$ ones followed by zeroes on the diagonal.
By assumption $\det(\sum_i A_iX_i)=\sum_j X_0^jP_j(X_1,\dots,X_{m-1})$ is non-zero.
Then, $P_r=\det(\sum_{i\geq 1} B_iX_i)$, and is also non-zero.
By our choice of $x_i$, that is $\det(\sum_{i\geq 1} B_ix_i)\neq 0$, we have that $\det(A_0X+\sum_{i\geq 1} A_ix_i)=\sum_j X_0^jP_j(x_1,\dots,x_{m-1})$ is a non-zero univariate polynomial of degree at most $n$, so that it has at most $n$ roots, which guarantees the algorithm will find a solution.
\end{proof}

\begin{theorem}\label{thm:reduc}
Given $(I,r)$ in a poly-reducible group of determinant $D$, we can find in polynomial time $x\in E^*$ and a basis $C$ of $I/x$, such that $h(C)+h(r/N(x))=(h(\Ord)+h(D))^{O(1)}$.
Also, $\Ord \subset I$.
\end{theorem}
\begin{proof}
Without loss of generality, we can assume that $I$ is an integral ideal, of basis $A=(e_i)_{i<n}$.
We then compute the Gram matrix corresponding to $x\mapsto \tr(N(x)/r)$, that is $G_{i,j}=(\tr(N(e_i+e_j)/r)-\tr(N(e_i-e_j)/r))/4$, which we can do in polynomial time.
We use~\autoref{thm:LLLred} to compute $U$ such that $U^tGU$ is bounded by $n2^nD$.
Now, we compute $B=AU=(b_i)_{i<n}$, and use~\autoref{lem:invert} with the multiplication matrices of $b_i$, to find $y=\sum_i x_i b_i$ invertible.
Finally, we return $y$ and $B/y$.

The running time is clear.
Because $II^{-1}=\Ord$, there exists $z\in \Z^n$ such that $Az$ is invertible, so that the condition of~\autoref{lem:invert} is fulfilled.

We now have $\tr(N(y)/r)=\sum x_i^2\tr(N(x_i)/r)\leq n^42^nD$.
Since $h(\Ord)\geq n$ and $y\in I$ implies $I/y \subset \Ord$, $h(r/N(y))=(h(\Ord)+h(D))^{O(1)}$ follows from~\autoref{lem:polyred}.
Also, for all $i$, we have $\tr(N(b_i/y)N(y)/r)\leq n2^nD$ so that $h(N(b_i/y)N(y)/r)=(h(\Ord)+h(D))^{O(1)}$.
Now we can compute $N(b_i/y)$ from $N(b_i/y)N(y)/r$ and $r/N(y)$, and hence $\tr(N(b_i/y))$ in time $(h(\Ord)+h(D))^{O(1)}$.
It implies $\log \tr(N(b_i/y))=(h(\Ord)+h(D))^{O(1)}$.
Then, with $H=L^tL$ the Gram matrix of $x\mapsto \tr(N(x))$ over $\Ord$ and its Cholesky decomposition, we have $\log (L(b_i/y)_j)=(h(\Ord)+h(D))^{O(1)}$ for all $j$.
Since $L_{j,j}\geq ||H||^{-n/2}$ we have $L^{-1}_{j,j}\leq ||H||^{n/2}$, we deduce $h(b_i/y)=(h(\Ord)+h(D))^{O(1)}$.
\end{proof}

\subsection{Powering} \label{sub:power}

\begin{theorem}\label{thm:power}
Given $(I,r)=((v),N(v))$ a principal polarized ideal and an integer $e$, we can compute $(I,r,1)^e$ over the compactification of the principal polarized ideal group in polynomial time.
Furthermore, the norm which is outputted have a height $h(\Ord)^{O(1)}$ and the ideal contains $\Ord$.
\end{theorem}
\begin{proof}
If $e=0$, we return $(\Ord,N(1),1)$.
If $e$ is even, we recursively compute $(I,r,1)^{e/2}=(K,u,s)$, use~\autoref{thm:reduc} with $(K^2,u^2)$ which returns an ideal $C$ and $x\in E^*$, and outputs $(C,u^2/N(x),s^2x)$
Else, we recursively compute $(I,r,1)^{e-1}=(K,u,s)$, use~\autoref{thm:reduc} with $(KI,ur)$ which returns an ideal $C$ and $x\in E^*$, and outputs $(C,ur/N(x),sx)$.
If at any point, the height is too large with respect to the bounds given by~\autoref{thm:reduc}, we fail.

Correctness is clear.
By induction, the output is reduced so its height without the formal product is bounded by $h(\Ord)^{O(1)}$.
It implies that the height of the bases in the formal product is bounded by $h(\Ord)^{O(1)}$, while the exponents are bounded by $e$ and the cardinal of the support is bounded by $O(\log(e))$.
Therefore, the algorithm runs in polynomial time.
\end{proof}
Note that we can, in fact, compute any circuit over the compactification of any poly-reducible group in polynomial time.
In particular, we may use shorter addition chains.

\subsection{Recovery of a high power of $v$} \label{sub:highpow}

We fix in this subsection a maximal ideal $\Frk{m}$ of $E$, and suppose a $\Q$-basis is given.
Remember that $E/\Frk{m}$ is a number field, which we denote $F$.
We define $\Ord'$ as a suborder of $\Ord/\Frk{m}$.

\begin{lemma}\label{lem:idealsvp}
Given an invertible integral ideal $\Frk{a}$ of $\Ord'$, we have for all $x\in \Frk{a}^k$, $x\neq 0$, $h(x)/n \geq k\log(N(\Frk{a}))-O(h(\Ord'))$.
\end{lemma}
\begin{proof}
It is well known that the norm is multiplicative for invertible ideals, see for example~\cite[Proposition 4.6.8]{cohen2013course}.
Also, for $x\in \Frk{a}^k$, $x\neq 0$, we have $(x)\subset \Frk{a}^k$ so that $N((x))\geq N(\Frk{a})^k$.
Then, $N((x))$ is also the absolute value of the determinant of the multiplication by $x$.
Using the Hadamard bound, we have $\log(N((x)))\leq n(\log(||x||)+h(\Ord'))$.
\end{proof}

\begin{lemma}\label{lem:lllbabai}
Given an invertible matrix $A$ with $\lambda_1=\min_{x\in \Z^n-\{0\}} ||Ax||$, and $c$ such that there exists $y\in \Z^n$ with $||Ay-c||\leq 2^{-n}\lambda_1$, we can recover $y$ in polynomial time.
\end{lemma}
\begin{proof}
This was proven by Babai~\cite[Theorem 3.1]{babai1986lovasz}, as an application of LLL.
\end{proof}

\begin{lemma}
Given a prime ideal $\Frk{p}$ of $\Ord'$ with a prime number $p\in \Frk{p}$, we have $x^{(N(\Frk{p})-1)p^k} \in 1+\Frk{p}^{k+1}$ for any invertible $x\in \Ord'-\Frk{p}$ and $k$ a positive integer.
\end{lemma}
\begin{proof}
We use induction on $k$.
$\Ord'/\Frk{p}$ is a field, so $x^{N(\Frk{p})-1}\in 1+\Frk{p}$ for any invertible $x\in \Ord'-\Frk{p}$.
Now, let $y\in \Frk{p}^k$.
If we develop $(1+y)^p-1-y^p$, $p\in \Frk{p}$ divides all binomial coefficients and the power of $y$ is at least one, so $(1+y)^p-1-y^p\in (p)\Frk{p}^k \subset \Frk{p}^{k+1}$.
Since $p\geq 2$, we also have $y^p \in \Frk{p}^{k+1}$ and hence $(1+y)^p \in 1+\Frk{p}^{k+1}$.
\end{proof}

\begin{theorem}\label{thm:retrouvepuiss}
Given $(I,r)=((v),N(v))$ a principal polarized ideal and $\Frk{p}$ be an invertible prime ideal of $\Ord'$ with $p\in \Frk{p}$ a prime integer and $v\not\in \Frk{p}+\Frk{m}$.
Then, we can output $k$ and $s$ in polynomial time such that $v^{(N(\Frk{p})-1)p^k}=[s]$ modulo $\Frk{m}$.
\end{theorem}
\begin{proof}
Let $e=(N(\Frk{p})-1)p^k$ for some integer $k$.
We first compute using~\autoref{thm:power} $(I,r,1)^e=(J,a,s)$.
We know that $1/a=N([s]/v^e)$, $h(1/a)=h(\Ord)^{O(1)}$ and $J^{-1}$ is integral, and therefore $[s]/v^e$ is an invertible integer, whose height is in $h(\Ord)^{O(1)}$.
Using~\autoref{lem:idealsvp}, there exists a constant $c$ such that with $k=h(\Ord)^c$, any non zero element in $\Frk{p}^{k+1}$ has a coordinate larger than $2^n$ times any coordinate of $[s]/v^e$.

Remark that we can compute in polynomial time $\Frk{p}^{k+1}$.
We then run~\autoref{lem:lllbabai} with the basis of $\Frk{p}^{k+1}$ and $[s]$ modulo $\Frk{m}$ and $\Frk{p}^{k+1}$ ; we call the result lifted to $\Ord$, $c$.
Because of the previous lemma, since $v\not\in \Frk{p}+\Frk{m}$, $v^e=1$ modulo $\Frk{m}$ and $\Frk{p}^{k+1}$.
Hence, $c$ and $[s]/v^e$ differs by an element of $\Frk{p}^{k+1}$ modulo $\Frk{m}$ but by definition of $k$, \autoref{lem:idealsvp} and~\autoref{lem:lllbabai}, it must be zero.
Therefore, we return $s/c$.
\end{proof}
Though $v$ may not be unique, the given power is.

\subsection{Recovery of $v$ over a field} \label{sub:field}

\begin{lemma}\label{lem:calcord}
Given an order $\Ord'$ over a number field, we can compute $\Ord' \subset \Ord'_p$ and integral invertible ideals $\Frk{p}_i$ of $\Ord'_p$ such that
\[ p\Ord'_p=\prod_i \Frk{p}_i^{e_i} \]
in time which is polynomial in the size of the input, and $p$.
Using randomness, we can do the same in polynomial time.
\end{lemma}
\begin{proof}
See~\cite[Sections 6.1 and 6.2]{cohen2013course}, where the only randomness used is for factoring polynomials modulo $p$.
\end{proof}

\begin{lemma}\label{lem:pgcd}
Let $e_i$ be integers, $s_i$ be formal products such that $[s_i]=v^{e_i}$ modulo $\Frk{m}$ and $I=(v)$.
Let $g$ be the greatest common divisor of the $e_i$.
Then, we can find in time polynomial in the size of the input and $g$, an element $w$ such that $w/v$ reduced modulo $\Frk{m}$ is a root of unity.
\end{lemma}
\begin{proof}
We can compute in polynomial time by applying a Hermite normal form algorithm over $e$, a vector of integers $u_i$ such that $\sum_i e_iu_i=g$.
We search for some prime $p$ such that the bases in the support of all $s_i$ are invertible modulo $p$.
Since $a\in \Ord$ is invertible modulo $p$ if and only if its norm is divisible by $p$, there exists a $p$ which is bounded by a polynomial of the size of the input which works and we can find it in polynomial time.
We then compute $s=\prod_i s_i^{u_i}$, and evaluate $[s]$ modulo $\Frk{m}$ and some sufficiently high power of $p$ (but polynomial in $h(I)^g$), so that we recover $y$, congruent to $v^g$ modulo $\Frk{m}$.
Finally, we factor $X^g-y$ in polynomial time over $F$ (see~\cite[Section 3.6.2]{cohen2013course}), and if there exists a linear factor $X-w$, we output $w$.
Else, we fail.

Remark that $v$ reduced modulo $\Frk{m}$ is a root of $X^g-y$, and the quotient of two roots must be a $g$-th root of unity.
\end{proof}

\begin{lemma}\label{lem:cyclo}
If $F$ is a cyclotomic field, we can choose in polynomial time two primes such that the gcd of the corresponding exponents in~\autoref{thm:retrouvepuiss} is polynomial.
\end{lemma}
\begin{proof}
If $F=\Q[\zeta_m]$, then we choose the first two primes which split in linear factors, which is equivalent to being congruent to one modulo $m$.
We use the latest version of Linnik's theorem, which says that the smallest prime congruent to $a$ modulo $k$ is $O(k^5)$~\cite[Theorem 2.1]{xylouris2011nullstellen}.
Then, the smallest prime congruent to one mod $m$, $p$ verifies $O(m^5)$.
Let $r$ be the smallest prime which does not divide $p-1$ or $m$, we know that $r=O(\log(m))$.
We then define $a$ to be the element of $\Z/(rm)$ congruent to one modulo $m$ and to $1+p$ modulo $r$.
Thus, we can define $q$ to be the smallest prime congruent to $a$ modulo $rm$, and $q=O((m\log m)^5)$, $q>p$.
Finally, for any $\alpha,\beta$, $\gcd((p-1)p^\alpha,(q-1)q^\beta)=\gcd((p-1)p^\alpha,q-1)<q$.
\end{proof}
Remark that we can make the gcd equal to $m$ using the technique of~\cite[Proposition 4.5]{C:LenSil14}, but the exponent then becomes 50.

\begin{lemma}
Let $H=F[\zeta_m]$ be a Galois extension of $F$ and $m\not\in \Frk{p}$ be some invertible prime ideal of $\Ord'_p$.
Then, $m|N(\Frk{p})-1$ if and only if $\Frk{p}\Z_K$ splits completely over $H$.
\end{lemma}
\begin{proof}
Using the properties of the conductor ideal, it is a standard fact that without loss of generality, we can assume $\Ord'_p=\Z_K$.
Also, $\Frk{p}$ does not ramify.
Remark that for any prime $\Frk{q}$ above $\Frk{p}$, we have $x^{N(\Frk{q})}=x$ over $\Z_H/\Frk{q}$.
Hence, $\Frk{p}$ splits completely is equivalent to $x^{N(\Frk{p})}$ fixes $\Z_H/\Frk{p}$.
But $\zeta_m \not\in\Frk{p}$ so that it is equivalent to fixing $\zeta_m$, which is $m|N(\Frk{p})-1$.
\end{proof}

\begin{theorem}\label{thm:field}
Given a polarized principal ideal $(I,r)=((v),N(v))$ and $\Frk{m}$, if $F$ is a cyclotomic field, or GRH is true, or we have access to randomness, we can find $w$ such that $v/w$ modulo $\Frk{m}$ is a root of unity in polynomial time.
We can also do this unconditionnally in time polynomial in $2^{n^{(1+o(1))/\log\log n}}$ and the size of the input.
\end{theorem}
\begin{proof}
The algorithm consists in applying~\autoref{lem:calcord} to generate the input of~\autoref{thm:retrouvepuiss}, and we combine the outputs using~\autoref{lem:pgcd}.
The crux of the matter is to bound the greatest common divisor of the exponents used.
The case of cyclotomic field is easily treated with~\autoref{lem:cyclo}.
Indeed, for each prime $p$, then either the image of $v$ is divisible by $p$ so the factor can then be removed and this happens at most a polynomial number of times, or we can use some $\Frk{p}$ above $p$.

Under GRH, we show that using all prime ideals $\Frk{p}$ of inertia degree one above all primes $p$ smaller than a polynomial will work.
Using the previous lemma and~\cite[Théorème 4]{serre1981quelques}, we have that there are a polynomial number of primes smaller than some polynomial who have a prime ideal above it of inertia degree one.
Hence, we can find two amongst them which does not divide $v$, and $m$ will be $\gcd((p-1)p^f,(q-1)q^e)<q$ if $q>p$.
Therefore, the algorithm runs in polynomial time.

For the other algorithms, we first start by trying for all the $2n^2+n+1$ smallest primes any ideal above it such that the image of $v$ is not in the ideal.
Each time a prime is detected as dividing the image of $v$, it can be factored out.
Let $m$ be the current greatest common divisor of the exponent used.
Remark that there are $n$ inertia degree possible, so that by the pigeonhole principle, there exist one degree $d$ with $2n+2$ corresponding primes $p_i$.
Let $p^k\mid m$ with $k$ positive and $p$ prime.
By removing $p$ from the list of $p_i$, we have that for $2n+1$ distinct primes $p_i$, $p^k\mid p_i^d-1$.
This implies that either $p^k=O(n^2\log n)$, or there are $2n+1$ elements of order dividing $d$ in $(\Z/(p^k))^*$.
If $p$ is odd, the group is cyclic so that $2n+1\leq d$ which is absurd.
Else $p=2$, the group has at most $2d$ elements of order dividing $d$, and $2n+1\leq 2d$ which is also absurd.
Hence at this point, $p^k=O(n^2\log n)$ for all $p^k\mid m$.

For the unconditional algorithm, we continue to do so for the first primes.
Let $P$ be the largest prime used, which we will fix later to some function in $2^{n^{\Theta(1/\log \log n)}}$.
Using the same argument, we have that for the new $m$ and some $d$, that either $m\leq O(Pn^2\log n)$ or there exist a $m'\mid m$ with $m'=O(Pn^2\log n)$ and odd such that all primes considered except possibly $O(n^2\log n)$ of them are of order dividing $d$ in $(\Z/(m'))^*$, and these primes are distinct elements modulo $m'$.
Thus, there are $\Omega(P/n/\log P)$ elements of order dividing $d$ in $(\Z/(m'))^*$, which is a proportion of $\Omega(1/n^3/\log^3(n)n^{-\Theta(1/\log \log n)})=\Omega(1/n^4)$.
Decomposing $(\Z/(m'))^*$ as a product of cyclic groups, we first consider the groups where all elements are of order dividing $d$, that is $(p-1)p^k\mid d$, for $(p-1)p^k \mid m'$ with $p$ prime.
Let $(q-1)q^r$ be another cyclic group of equal order.
Then, without loss of generality $q\geq p$ and if $k$ is positive, $q\mid (p-1)p^k$ which is not possible.
Hence $q$ is unique with respect to $(p,k)$ and using Wiegert's theorem, we deduce that there are at most $d^{(1+o(1))/\log \log d}$ cyclic groups where all elements are of order at most $d$.
But because of Lagrange's theorem, there are at most $O(\log(n))$ cyclic groups where not all elements are of order dividing $d$.
Hence, for \[ P\geq (n^2\log n)^{O(\log(n))+d^{(1+o(1))/\log \log d}}\]
this is absurd.
The total running time is therefore in $2^{n^{(1+o(1))/\log\log n}}$.

Using randomness, we sample a polynomial number of integers smaller than $B^{1/f}$ with $B=\exp(O(\log(d)\log(\log d)\log(\log \log d)))$ where $d=d_{F[\zeta_q]}$ for some prime power $q$ dividing $m$ such that $F$ has no $q$-th root of unity.
If the number divides $m$, which happens with negligible probability, we restart.
Else, if the number is prime, we use all prime ideals above this prime.
The probability that some ideal divides $v$ is negligible.
We use all $f$ from one to $n$.

\cite[Théorème 3]{serre1981quelques} shows that the likelihood of a complete split is then $\epsilon\leq c/[F[\zeta_q]:F]$ for a uniform prime of norm below $B$ and some universal constant $c$.
Assume $\epsilon<1/2$. Then, there are $k\geq 1$ inertia degrees $f$ such that the likelihood for a prime below $B^{1/f}$ to have prime above it with an inertia degree $f$ is at least $(1-\epsilon)/n$.
Further, the likelihood that a uniform prime of norm below $B$ with one of these inertia degrees not to completely split is at least $(1-\epsilon)k/n$.
Therefore, there exist an inertia degree $f$ such that a uniform prime of norm below $B$ of inertia degree $f$ will not completely split with probability at least $(1-\epsilon)/n$, and a uniform prime below $B^{1/f}$ has a probability at least $(1-\epsilon)/n$ to have a prime of inertia degree $f$.
We deduce that the above procedure takes polynomial time to find a prime ideal which does not split completely with high probability.
In case we find a prime not have a complete split, because $p\nmid m$, we have $\gcd(m,(p^f-1)p^e)\leq m/2$ so that with high probability, after a polynomial number of tests, we have that $p^k|m$ only if $F$ has a $p^{k/c}$ root of unity.
We deduce then that $m=O(n\log \log n)^c$, so that the running time is polynomial.
\end{proof}

\subsection{Recovery of $v$} \label{sub:final}
\begin{theorem}
Given $E$, we can compute $\Frk{m_i}$ such that $E$ is the sum of a nilpotent vector space and $\prod_i E/\Frk{m_i}$ where $\Frk{m_i}$ are maximal ideals of $E$ in polynomial time.
We can also compute the generators of group of roots of unity of $E$ and $\Ord$ in polynomial time and $e_i$ such that $\Ord=\prod_i e_i\Ord$ and $e_i\Ord$ has only trivial idempotents.
\end{theorem}
\begin{proof}
See~\cite[Theorem 1.2]{lenstra2015algorithms} for the first algorithm, which starts by expressing the product of number fields as $\Q[x]/(f)$ and then compute the factorization of $f$ by LLL.
The second part is quite involved and is proven in~\cite{lenstra2015roots}.
\end{proof}

The following method for recovering the root of unity is heavily inspired by~\cite{C:LenSil14}.
\begin{lemma}\label{lem:cmunite}
Let $w\in \mu(\Ord)$ where $\Ord$ is a CM order.
Then $w\barre w=1$.
\end{lemma}
\begin{proof}
With~\autoref{lem:cmnilpot} and~\autoref{thm:structalg}, $E$ is a product of fields.
Hence, $E\otimes \C$ is isomorphic to $\C^n$, and the conjugation can be projected to a conjugation over $\C$ which has the same properties.
Therefore, it is the standard conjugation over $\C$, and all roots of unity $\zeta$ over $\C$ verifies $\zeta\barre \zeta=1$.
Since the projection of a root of unity $w\in \Ord$ is a root of unity, we have $w\barre w=1$.
\end{proof}

\begin{lemma}
For any $a\in \Ord$ and $\Ord$ a CM order, $\tr(a\barre a)\geq r$ where $r$ is the dimension of $a\barre a \Ord$.
\end{lemma}
\begin{proof}
Without loss of generality, $a\neq 0$.
Consider the application $x\mapsto xa\barre a$ over $a\barre a \Ord$.
Its determinant is a non-zero integer since there is $a\barre a$ is not nilpotent, so that the inequality of arithmetic and geometric means over the eigenvalues gives the result.
\end{proof}

\begin{lemma}
If $A$ and $B$ are CM orders, then $A \tens B$ is a CM order and if $\sum_i a_i \tens b_i \in \mu(A \tens B)$, we have $A=\prod_i a_i\barre{a_i}A\tens b_i\barre{b_i}B$.
\end{lemma}
\begin{proof}
The only difficult point in the first statement is to show that $x\mapsto \tr(x\barre x)$ is positive-definite.
This comes from the fact that the corresponding Gram matrix is the Kronecker product of the two Gram matrices corresponding to $A$ and $B$, which can be diagonalized thanks to the spectral theorem.

Then, if $\sum_i a_i \tens b_i \in \mu(A \tens B)$ where the sum is finite and $a_i,b_i\neq 0$, we have $(\sum_i a_i \tens b_i)(\sum_i \barre{a_i} \tens \barre{b_i})=1$ with~\autoref{lem:cmunite}.
Therefore, $\sum_i a_i\barre{a_i} \tens b_i\barre{b_i}=1$.
We deduce $\sum_i \tr(a_i\barre{a_i})\tr(b_i\barre{b_i})=\tr(1)$ and using the previous lemma, $A\tens B=\prod_i a_i\barre{a_i}A\tens b_i\barre{b_i}B$ as product of suborders.
\end{proof}

\begin{lemma}
$B=\Z[X]/(X^n-1)$ equipped with $X \mapsto X^{n-1}$ is a CM order and $\mu(B)$ is generated by $X$ and $-1$.
Its idempotents are zero and one.
\end{lemma}
\begin{proof}
Remark that for any $\omega \in \mu(B)$, we have with~\autoref{lem:cmunite} $\omega\barre \omega=1$.
Then, with $\omega=\sum_{i=0}^{n-1} a_iX^i$, $\tr(\omega\barre \omega)=n\sum_{i=0}^{n-1} a_i^2$.
Therefore, $\omega=\pm X^i$ and the converse is clear.

If $e$ is an idempotent, then $e\barre e$ is also an idempotent.
But if $e\neq 0$, $\tr(e\barre e)\geq n$ so that $e\barre e=1$.
Hence $e$ is invertible, so that $e=1$.
\end{proof}

\begin{theorem}
Given a CM order $\Ord$, an ideal $I=\omega\Ord$ with $\omega\in\mu(\Q \tens \Ord)$, we can find $\zeta\in \mu(\Q \tens \Ord)$ such that $\omega\zeta\in \Ord$ in polynomial time.
\end{theorem}
\begin{proof}
Without loss of generality, we can assume $\omega \in \mu(\Q \tens \Ord)_p$ and we know some $e\leq 2n$ such that $\omega^e=1$.
We now compute all the primitive idempotents $e_i$ of $\Ord$ and by combining the results for all $I/e_i\Ord$ over $e_i\Ord$, we can assume $\Ord$ has only trivial idempotents.

We then build the CM order $\Ord \tens \Z[X]/(X^e-1)$, by concatenating the basis of $I^i=(\omega^i)$.
We now use~\cite[Theorem 1.2]{lenstra2015roots} to find the generators of the roots of unity of this order.
Because of the previous lemmata, they are of the form $w \tens X^i$ with $w \in \mu(\Ord)$.
By combining the generators, we can deduce a root of unity of the form $w \tens X$.
Hence, $w \in \omega\mu(\Ord)$ so we can output $1/w$.
\end{proof}

\begin{theorem}
Given $(I,r)$, if $E'$ is a product of cyclotomic field, or GRH is true, or we have access to randomness, we can find $v\omega\in \mu(E)\Ord$ such that $I=(v)$ or prove there is no such $v\omega$ with $N(v)=r$.
We can also do this unconditionnally in time $2^{n^{(1+o(1))/\log\log n}}$.
Furthermore, if $\Ord$ is a CM order, we can find $v\in \Ord$ such that $(I,r)=((v),N(v))$.
\end{theorem}
\begin{proof}
We first compute all $\Frk{m}_i$, apply~\autoref{thm:field} for each $\Frk{m}_i$ and recover some $x$ using the Chinese remainder theorem.
Then, we compute $J=I/x=(\omega+a)$ where $a^n=0$ and $\omega \in E$ is a root of unity.
From the knowledge of the group of roots of unity of $E$, we can deduce in polynomial time $e$ such that $\omega^e=1$.
We may then compute $J^{e2^k}=(1+b)^{2^k}$ with $b=(\omega+a)^e-1$, which is easily seen to be a nilpotent.
Since $1-x\mapsto \sum_{i=1}^n -x^i/i$ for $x$ nilpotent is a morphism, whose inverse is $x\mapsto \sum_{i=0}^n x^k/k!$ (see~\cite[Proposition 8.1]{lenstra2015algorithms}), this takes time which is polynomial in $k\log(e)$.

We can compute $\Ord'$, the largest order which contains $\Ord$ and all roots of unity of $E$ in polynomial time, and we have $b\in \Ord'$.
Therefore, we have $(1+b)^2=1+b^2$ in $\Ord'/(2)$ so that $(1+b)^{2^{\lceil \log n \rceil}} \in 1+2\Ord'$.
We deduce $(1+b)^{2^k}=1+2^{k-\lfloor \log n \rfloor}\Ord'$.
We apply~\autoref{thm:reduc} to $J^{e2^k}$ to produce an invertible $y$ such that $y/(1+b)^{2^k}\in \Ord'$, with $h(y/(1+b)^{2^k})=h(\Ord)^{O(1)}$.
Hence, we choose $k$ sufficiently large so that $y/(1+b)^{2^k}$ is equal to the lift of $y$ modulo $2^{k-\lfloor \log n \rfloor}\Ord'$.
We can therefore compute $(1+b)^{2^k}$ and using the two morphisms, deduce $z=x(1+a/\omega)$.

We finally apply the previous theorem with $I/(z)$ to recover $\omega$.

\end{proof}

\subsection{Comments on the algorithm}

One problem with the given algorithm is that it is \textit{not explicit}.
In particular, we need an upper-bound on the constants of the algorithms dealing with the norms ($\tr,N$, multiplication and inversion in $\Cal{N}$).
However, this seems to be an unavoidable consequence of our black-box model, as slower algorithms mean a possibly larger set of solutions.
Also, for any application exposed at the beginning, these constants are explicit.
Furthermore, if we impose that there is a solution, one can simply increase the constant until we reach the solution.

Another difficulty is the sheer complexity of the algorithm, both in term of code length and running time.
However, a large part of this complexity can be removed.
Indeed, in practice, as soon as we combine information given by a couple of exponents (typically two, if we manage to find small primes having a prime ideal of inertia degree one above them), the greatest common divisor becomes tiny.
It can be explained by a heuristic application of Chebotarev's theorem~: if $p^k$ divides the current greatest common divisor but does not divide the number of roots, the probability that $k$ will not decrease is $O(1/p)$.
Hence, we can simply ignore all primes $p$ where $p\Ord$ is not invertible, or $p$ divides the discriminant of the polynomial ; and beside the exponentiation, we only need to factor the polynomial defining the number field to produce the prime ideals.
Also, applications can generally cope with finding the solution up to root of unity of $E$, since they usually work with an order in a number field with a known polynomial, which contains few roots of unity, and no nilpotents.
Root extraction can be efficiently computed if we know an inert prime ; a Newton-Hensel iteration may also work.
Ideal multiplication can be accelerated by compressing the lattice, see~\cite[Section 4]{chen2005blas}.
While the exponent needed might seem to be huge, it is usually fairly small.
For example, when $\Ord=\Z[\zeta_m]$, the precision needed is exactly the size of a typical LLL reduction of a lattice of determinant one, which is in practice $1.022^n$ in dimension $n$~\cite{EC:GamNgu08}.
Finally, we explain in~\autoref{subsec:exposant} that under plausible heuristics, the exponent is bounded by $O(\log(h(\Ord))\log(\log(h(\Ord))))$, so that the resulting complexity is in general a couple of lattice reductions.
We add that we can save an ideal powering using a Hensel iteration :
\begin{lemma}~\label{lem:racine}
Let $s$ be a formal product such that $[s]=v^e$ for some known $e$, and $v\in \Ord'$.
Given a bound on $h(v)$ and a prime invertible ideal $\Frk{p}$ of inertia degree $f$ above the prime $p$, with $d=\gcd(e,p^f-1)$, $p\nmid e$ and $v\not\in \Frk{p}$, we can compute $v$ in time polynomial in the size of the input, $p$ and $d$.
\end{lemma}
\begin{proof}
We select some $k$ and compute $[s] \Mod \Frk{p}^k$ and $c$ the inverse of $e/d$ modulo $p^f-1$.
We can then factor $X^d-[s]$ in the finite field $\Ord'/\Frk{p}$ in time polynomial in $p$, $d$ and the size of the input.
For each root $r$, we have $r^c$ a root of $X^e-[s] \Mod \Frk{p}$, which we can extend (since $p \nmid e$) using Hensel lifting to a root of $X^e-[s] \Mod \Frk{p}^k$.
Using~\autoref{lem:idealsvp} and~\autoref{lem:lllbabai}, for some polynomially large $k$, we can recover $v$ from $v \Mod \Frk{p}^k$.
\end{proof}

\section{Applications}\label{sec:appli}

\subsection{Dimension halving}

In this subsection we define $\Ord$ to be a CM order.
The algorithm was first evoked in Gentry's dissertation~\cite[Section 6.2]{gentry2009fully} before being developped in GGH~\cite[Section 8.8.1]{EC:GarGenHal13}.
We correct here two benign mistakes in the algorithm.
The first is that we should prove the existence of a short \textit{non-zero} vector.
The second is that the Gentry-Szydlo does not give a \textit{unique} solution.

\begin{lemma}\label{lem:reddim}
Given an integral invertible ideal $I=(v)$ of $\Ord$ and some $x\in I-\{0\}$ which minimizes $\tr(x\barre x)$, then there exists $y\neq 0$ in $v \Ord^+$ or $v\Ord^-$ such that $0<\tr(y\barre y)\leq 2\tr(x\barre x)$.
Furthermore, these two lattices are included in $I$.
\end{lemma}
\begin{proof}
Let $z=x/v$.
Remark that $\tr(v\barre z\barre{v\barre z})=\tr(vz\barre{vz})=\tr(x\barre x)$.
Since $x\mapsto \tr(x\barre x)$ is a quadratic form, there exists $s\in \{-1,0,1\}$ such that $0<\tr(v(z+s\barre z)\barre{v(z+s\barre z)})\leq 2\tr(x\barre x)$ and $s=0$ only if $z=\barre z$.
If $s=0$, then $vz \in v\Ord^+$.
If $s=1$, then $v(z+\barre z) \in v\Ord^+$.
Else $s=-1$ and $v(z-\barre z) \in v\Ord^-$.
\end{proof}

\begin{theorem}
Given a CM order $\Ord$ included in a product of $k$ fields a principal ideal $I$, using one call to~\autoref{thm:gsalg}, having access to an oracle finding a non-zero vector in a lattice at most $\gamma$ times larger than the shortest non-zero vector, and time polynomial in the size of the input and $2^k$, we can find a non-zero vector at most $\gamma\sqrt{2}$ larger than the shortest non-zero vector of the ideal.
Furthermore, all calls to the oracle are of dimension at most $\max(\dim \Ord^+,\dim \Ord^-)$.
\end{theorem}
\begin{proof}
Without loss of generality, $I$ is invertible and integral.
Let $v\in \Ord$ such that $I=(v)$.
We first compute $\barre I/I$ and run~\autoref{thm:gsalg}, which returns some $\omega \barre v/v$ and $\omega \in \mu(\Ord)$.
Then, for all $\zeta \in \mu(\Ord)/\mu(\Ord)^2$, we deduce $J=I(1+\zeta \omega \barre v/v)=(v+\zeta \omega \barre v)$.
For some $\zeta$, we will have $\zeta \omega=\barre{w}^2$ and $w\in \mu(\Ord)$.
Thus with~\autoref{lem:cmunite}, $J=(vw+\barre{vw})$ and $I=(vw)$.
We can then compute $J\cap \Ord^+=(vw+\barre{vw})\Ord^+$ since $vw+\barre{vw}\in \Ord^+$.
Dividing by $1+\barre{vw^2}/v$, we get a basis of $v\Ord^+$.
Now, the direct sum of $v\Ord^+$ and $v\Ord^-$ is $2I$ so we can compute a basis of $v\Ord^-$.
Using~\autoref{lem:reddim}, we just need to call the oracle on these two lattices.

The complexity is given by the fact that there are at most $2^k$ different $\zeta$.
\end{proof}
Usually $k=1$ so that the algorithm is efficient.
Then, either $\Ord^+=\Ord$ and nothing happens or $\dim \Ord^+=\frac{1}{2} \dim \Ord$ and the algorithm halve the dimension for a moderate cost.

It is easy to show that considering only $\zeta=1$ does not work.
For example, with $I=(1+i)\Z[i]$, we have $\barre I/I=-i\Z[i]=\Z[i]$ so that we may recover $1$ with Gentry-Szydlo's algorithm.
Then $I(1+1)=(2+2i)\Z[i]$ is not generated by an element of $\Z[i]^+=\Z$.

\subsection{Solving the norm equation}

\begin{problem}
We are given a CM order $\Ord$ and $r\in \Ord^+$.
We want to know all $x\in \Ord$ such that $x\barre{x}=r$.
\end{problem}

This is the norm equation problem, in the case of a CM order.
The following algorithm was introduced by Howgrave-Graham and Szydlo~\cite{howgrave2004method}.
See~\cite{simon2002solving} for a more general technique.
Remark that there may be many solutions since $(x\barre y)(\barre x y)=(xy)(\barre x \barre y)$, and possibly more than a polynomial.
Also, this case seems to show that in a way, we are factoring a number, and hence, discovering factors of its algebraic norm.
Hence, it is plausible that the following algorithm is close to optimal.

\begin{theorem}
Let $\Ord$ be a CM order over a number field.
Given the factorisation of the algebraic norm over $\Q$ of $r\in \Ord^+$, with $d$ the number of divisors and $r\in \Ord^+$, we can compute all $x\in \Ord$ such that $x\barre x=r$, in time polynomial in the size of the input, $d$ and calls to~\autoref{thm:gsalg}.
\end{theorem}
\begin{proof}
Without loss of generality, using~\cite[Section 6.1]{cohen2013course} we can assume that all primes involved are invertible in $\Ord$ and $\Ord^+$.
We may then find the factorisation of $r\Ord^+$ using~\cite[Sections 4.8.3 and 6.2]{cohen2013course}.
Then, $\Q \otimes \Ord$ is a Galois extension of $\Q \otimes \Ord^+$ so that a prime ideal $p$ of $\Ord^+$ is inert or factored into $\Frk{p}\Frk{\barre p}$.
If it is inert, then the $p$ valuation of $x\Ord$ must be half the $p$ valuation of $r\Ord^+$.
Else, the $\Frk{p}$ valuation of $x\Ord$ must be less than the $p$ valuation of $r\Ord^+$, and the $\Frk{\barre p}$ valuation is uniquely determined by it.
Therefore, there are at most $d$ $x\Ord$ distinct, and we can find all of them.
Using~\autoref{thm:gsalg}, we may then obtain $x$.
\end{proof}

\subsection{Lowering the exponent and applications}\label{subsec:exposant}

In many cases, the norm, traces and exponentiation are in fact smooth functions.
We can leverage this property by trying to run Gentry-Szydlo's algorithm with an approximate norm.
Indeed, what we need is that the last reduction in our powering algorithm (\autoref{thm:retrouvepuiss}) gives a meaningful result.
Of course, the quality of the approximation depends on the exponent used.
We show here that heuristically, we can use tiny exponents.
The idea comes from GGH~\cite[Section 8.6]{EC:GarGenHal13}.
Since all algorithms of this subsection use the following strong heuristic, we will also allow them to use randomness and GRH.

\begin{heuristic}
Let $F$ be a number field of $\Q$-dimension $n$ with exactly $m$ roots of unity.
Then, the expected value of the number of prime ideals above $p=am+1$ of inertia degree one is $\Omega(\frac{m\phi(m)}{\log(p)n})$ for a random $a$.
\end{heuristic}
\begin{proof}
Note that a prime $p$ having a prime ideal above it of inertia degree one must be of the form $am+1$.
The density of prime numbers among integers of this form is $\frac{m}{\phi(m)\log(p)}$.
$p$ always factors over $\Q[\zeta_m]$ into $\phi(m)$ ideals of inertia degree one.
The sum of $k$ times the density of prime ideals of $\Q[\zeta_m]$ having $k$ prime ideals above it of inertia degree one is, by Chebotarev theorem, is $r\phi(m)/n$ for some positive integer $r$, which is the average number of fixed point in the Galois group.
Assuming the two results occur somewhat randomly, and independently implies the heuristic.
\end{proof}

\begin{theorem}
Let $\Ord$ be a split CM order with no nilpotents beside zero.
If the heuristic assumption is true for each number fields, we can in polynomial time, given $(I,r)=(v,N(v))$, find some solution $v\in \Ord$ if it exists.
Furthermore, the exponent used in calls to~\autoref{thm:retrouvepuiss} is $e=h(\Ord)^{O(\log(h(\Ord)))}$.
Therefore, if $\tr(N(x)(\tilde{r}/r)^e)/\tr(N(x)) \in [1/2;2]$ for all $x\in \Ord-\{0\}$, we can find some solution $v\in \Ord$ to $(I,r)=(v,N(v))$ given $I$ and $\tilde{r}$.
\end{theorem}
\begin{proof}
Using~\autoref{thm:reduc}, we can assume that $h(v)=h(\Ord)^{O(1)}$.
Let $p_i$ be the sequence of prime numbers.
For some $k$, we let $e=m\prod_{i=0}^{k-1} p_i$.
We then proceed just like in~\autoref{thm:retrouvepuiss} with $(I,r,1)^{e}$.
Remark that there are at least $2^k$ divisors of $e$ of the form $km$, and $\log(e)=O(k\log k)$.
We therefore expect $\Omega(\frac{2^km\phi(m)}{nk\log k})$ of the $1+d$ with $d$ divisor of $e$ to be prime with a prime ideal $\Frk{p}_{a,j}$ of inertia degree one above it.
Only $O(\log(h(\Ord)))$ of these can divide $v$.
The inverse of the returned ideal is generated by a small integer in $1+\prod_j \Frk{p}_{j}$.
Hence, if the determinant of $\prod_j \Frk{p}_{a_j}$ is exponential in $h(\Ord)^{O(1)}$, we can proceed.
Then, using another prime of inertia degree one, we can finish in polynomial time with~\autoref{lem:racine}.
Our condition is then $\Omega(\frac{2^km\phi(m)}{n})=h(\Ord)^{O(1)}$, so that some $k=O(\log(h(\Ord)))$ works.
\end{proof}
In case our number field is $\Z[\zeta_n]$, we need only the product of the primes to be above the LLL approximation factor, which in practice is $\approx 1.022^{\phi(n)}$~\cite{EC:GamNgu08}.
For $n=2^{16}$, we can use $k=8$ so that $e=635678883840$ and the product is around $2^{1048}$ which is larger than the required $2^{1029}$.
This implies that only $\approx 50$ lattice reductions are needed.
The sum of $2/\ln(d)$ for all $n\mid d\mid e$ is $\approx 27$, and there are 38 $d+1$ which are primes.
Taking $k=21$ leads to 168076 primes instead of the predicted 98361, and the product has more than 10 million bits while $e<2^{112}$.

This theorem can be used to recover a unit $u$ from its approximation $\tilde{u}$ by calling it with $(\Ord,\tilde{u} \barre{\tilde{u}})$ within the corresponding split CM order.
If $\Ord$ is in a number field, $\R \otimes \Ord \simeq \R^r\C^s$ and by applying some complex logarithm, the image of $\Ord^*$ is a lattice of dimension $r+s-1$.
Now, the precision required in this basis is simpler to express : the error should be at most $n^{-O(\log(\log n))}$ on each coordinate of the image of $\tilde{u} \barre{\tilde{u}}$.

The following theorem can be seen as a way to compute a greatest common divisor.
\begin{theorem}
Let $\Dis$ be a samplable distribution over the number field $E$ such that for all embedding $\psi:\Ord\to \C$, $\log(|\psi(a)|)$ has standard deviation at most $\sigma$.
Given $(v)$ and $k$ samples $s_i=va_i$ where the $a_i$ are independent and sampled from $\Dis$, if the heuristic holds, we can recover $v\mu(\Ord)$ in polynomial time if $k\geq \sigma^2 n^{O(\log \log n)}$ except with negligible probability.
\end{theorem}
\begin{proof}
We fix an embedding of $E$ into $\C$ and then we can define the split CM order using the norm $x\mapsto x\barre{x}$.
We compute $\tilde{r}$ the average of $s_i\barre{s_i}$ divided by the average of $a_i \barre{a_i}$ computed by sampling from $\Dis$.
Using Chebyshev inequality, we can prove that if we use $\sigma^2 n^{O(\log \log n)}$ samples, then for all embedding $\psi:E\mapsto \C$, $$\psi(\tilde{r})/|\psi(v)| \in [1-n^{-O(\log \log n)};1+n^{-O(\log \log n)}]$$
with probability at least $1/2$.
Since $\tr(N(x))=\sum_{\psi} |\psi(x)|^2$, we can use the previous theorem.
\end{proof}
The original Gentry-Szydlo attack on NTRU signatures~\cite{EC:GenSzy02} is essentially an application of this theorem.
It improves on it by remarking that if $\Ord$ is a CM order we can compute $(v\barre v)$, reduce this basis, and use it to decode $\tilde{r}$ and recover $v\barre v$.
Another possibility which works for any order is to decode $\tilde{r}$ over a basis of $\Norm$ by truncating the coefficents, which has the advantage of being polynomial time.
It gives a proven algorithm which is polynomial, and needs a number of samples which is about the maximum coefficient of $v\barre v$.

Note that taking the ideal generated by all $s_i$ should get $(v)$ for most applications so that the hardest condition to achieve is the possibility of sampling from $\Dis$.

\begin{theorem}
Let $\Ord$ be in a number field, and for some embedding in $\C$, we define $N(x)=x\barre x$.

Given $(I,r)$ a polarized ideal of determinant $\Delta$, if the heuristic holds, we can determine if there exists a $v\in I$ such that for all embedding $\psi$ into $\C$, we have $|\psi(v)|^2/\psi(r) \leq 1+1/e$ and find it, for some $e=\log(h(\Ord))^{O(\log(h(\Ord)))}$.
\end{theorem}
\begin{proof}
We use $e=m\prod_{i=0}^{k-1} p_i$ and compute $(I,r,1)^e=(J,p,s)$.
Now, $\det(I^e)=\det(I)^e$ using~\cite[Proposition 4.6.8]{cohen2013course}, so that the determinant of $(I,r)^e$ is also $\Delta$.
Without loss of generality, we can assume $h(I)=h(\Ord)^{O(1)}$.
We select a random subset of half the prime ideals of inertia degree one above $p$ with $p-1 \mid e$, so that with high probability $v^e=1 \Mod K$ where $K \cap I$ has no non-zero vector shorter than $2^n3n$.
We deduce that $v^e/[s] \in J \cap (1/[s]+K)$ and $\tr(N(v^e/[s])/p)=\tr(N(v^e)/r^e)\leq 3n$.
Therefore, we can apply Babai's algorithm~\autoref{lem:lllbabai} on $J \cap (1/[s]+K)$ equipped with the norm $x\mapsto \tr(N(x)/p)$ and recover $v^e$ as a formal product.
We now find another small prime ideal and using~\autoref{lem:racine}, we recover $v$.
\end{proof}
Note that this implies that finding the shortest vector (for some norm) of invertible ideals is easy if it is almost as small as it can be ($1$).
If $r=N(v)$, then $\det((v))/\det(I)<2$ so that $I=(v)$ which is the standard case.

\section*{Acknowledgement}

We thank Pierre-Alain Fouque for his comments allowing to improve a draft of this paper.

\bibliography{cryptobib/crypto,ref}
\bibliographystyle{unsrt}

\appendix

\section{Exploiting roots of unity}

We assume here that $E$ is a number field with $m$ roots of unity in $\Ord$.
We show how their presence allows to accelerate standard lattice algorithms when the geometric norm is $||x||^2=\tr(N(x))=\tr(x\barre x)$.
We define $P$ as $X^{m/2}+1$ if $4\mid m$, $X^m-1$ else ; and let $m'$ be the degree of $P$.
Remark that there is a natural bijection between the roots of unity of $\Ord$ and $\pm X^i$ modulo $P$.

\begin{theorem}
For all $x\in I$, with $I$ an ideal of $\Ord$, we have $x\omega \in I$ and $||x||=||x\omega||$ for any $\omega \in \mu(\Ord)$.
\end{theorem}
\begin{proof}
The first property stems from $I$ being an ideal, the second from $N(\omega)=\omega\barre \omega=1$.
\end{proof}
This implies that $I$ has at least $|\mu(\Ord)|$ non-zero shortest vectors, making the extreme pruning algorithm~\cite{EC:GamNguReg10} about $|\mu(\Ord)|/2$ times faster than on a "random" lattice, since $x\mu(\Ord)$ is somewhat uniform over the sphere.

Also, sieving algorithms (see~\cite{hanrot2011algorithms,laarhovensearch}\footnote{Beware that the litterature often uses different way for expressing multiplication, multiplication by a root of unity or conjugation, such as (nega)cyclic matrices, rotation and reflex polynomial.} for surveys) can take advantage of this by reducing the size of the list of vectors by a factor of $|\mu(\Ord)|/2$ for the same reason.
A recent algorithm~\cite{EPRINT:BecLaa15} works by introducing a hash function $h$ which for a vector returns the index of the largest coordinate, as well as its sign.
It is then randomized to $h_a(x)=h(ax)$ for a Gaussian $a$ to produce a locality-sensitive hash function $H$ by concatenating outputs of several $h_a$.

We can improve on this by embedding $I$ and $E$ in $\Q[X]/(P(X))[Y]/(Q(Y))$ for some irreducible polynomial $Q\in \Q[\zeta_m][Y]$ of degree $n/\phi(m)$, so that they have the same geometry.
We can now choose $h_a(x)=h(ax)$ and observe that $h_a(x\omega)$ for $\omega$ a root of unity is simply a rotation of $h_a$.
Hence, we can build $H$ as the concatenation of $h_{a_0},h_{a_1},\dots,h_{a_k}$ where the output of $h_{a_0}$ is forced to be on a positive monomial of the form $Y^i$ by considering the unique root of unity which allows this.
The algorithm then has to compute the shortest element among $x+\omega y$ for all $\omega\in \mu(E)$.
We now show that this can be computed efficiently.

\begin{theorem}
Given $x,y\in \Q[X]/(P(X))[Y]/(Q(Y))$ with $P$ and $Q$ defined as above, we can compute $\argmin_{\omega \in \mu(\Ord)} ||x+\omega y||$ in $O((n/\phi(m))^2m\log m)$ arithmetic operations.
\end{theorem}
\begin{proof}
We denote $x=\sum_{i=0}^{n/\phi(m)-1} x_i Y^i$ for any $x$.
Now $\pr{x}{y}=\sum_{i,j} G_{i,j}\pr{x_i}{y_j}$ for some Gram matrix $G$ with the scalar product $\pr{x_i}{y_j}$ corresponding to the norm over the CM order.
Hence, we only need to show how to compute $\pr{a}{bX^i}$ for $a,b\in \Q[X]/(P(X))$ and all $i$ in $O(m\log m)$ operations.

Since the norm over $\Q[X]/(P(X))$ is $x\mapsto \tr(x\barre x)$, we have $\pr{a}{b\omega}=\tr(a\barre{b \omega})$.
We can therefore compute $a\barre{b}$ with a Fourier transform in time $O(m \log m)$.
Finally, $\tr(aX^{-i})$ is exactly the $i$-th coefficient of $a$.
\end{proof}
This implies an overall speed-up of $m^{1.43+o(1)}$, while~\cite{EPRINT:BecLaa15} gives a speed-up of only $O(m)$, and the ideals were required to be over a ring of the form $X^m\pm 1$.
The use of Fourier transform for accelerating geometric computations was first introduced by~\cite{EPRINT:BosNaePol14}.

\end{document}